\documentclass[12pt]{article}
\usepackage{amsthm,amsfonts,amsmath,amssymb,mathrsfs,tikz,url,hyperref}
\usepackage{footmisc}
\usepackage{authblk}
\usepackage{mathtools}
\usepackage{physics}
\usepackage{dsfont}

\title{Fixed Points of Completely Positive Trace-Preserving Maps in Infinite Dimension}

\author{
Roderich Tumulka\footnote{Fachbereich Mathematik, Eberhard-Karls-Universit\"at T\"ubingen, Auf der Morgenstelle 10, 72076 T\"ubingen, Germany. E-mail: roderich.tumulka@uni-tuebingen.de } \, and 
Jonte Weixler\footnote{Fachbereich Mathematik, Eberhard-Karls-Universit\"at T\"ubingen, Auf der Morgenstelle 10, 72076 T\"ubingen, Germany. E-mail: jonte.weixler@student.uni-tuebingen.de}}
\date{November 20, 2024}

\newcommand{\be}{\begin{equation}}
	\newcommand{\ee}{\end{equation}}
\newcommand{\Hilbert}{\mathscr{H}}

\theoremstyle{plain}
\newtheorem{thm}{Theorem}
\theoremstyle{definition}

\newtheorem{lem}{Lemma}

\newtheorem{remark}{Remark}
\newtheorem{Corollary}{Corollary}

\begin{document}

\maketitle
\begin{abstract}
Completely positive trace-preserving maps $S$, also known as quantum channels, arise in quantum physics as a description of how the density operator $\rho$ of a system changes in a given time interval, allowing not only for unitary evolution but arbitrary operations including measurements or other interaction with an environment. It is known that if the Hilbert space $\Hilbert$ that $\rho$ acts on is finite-dimensional, then every $S$ must have a fixed point, i.e., a density operator $\rho_0$ with $S(\rho_0)=\rho_0$. In infinite dimension, $S$ need not have a fixed point in general. However, we prove here the existence of a fixed point under a certain additional assumption which is, roughly speaking, that $S$ leaves invariant a certain set of density operators with bounded ``cost'' of preparation. The proof is an application of the Schauder-Tychonoff fixed point theorem. Our motivation for this question comes from a proposal of Deutsch for how to define quantum theory in a space-time with closed timelike curves; our result supports the viability of Deutsch's proposal.	
\medskip
		
\end{abstract}

\section{Introduction}
\label{Introduction}
We present here a sufficient condition for the existence of fixed points of completely positive trace-preserving maps. Completely positive trace-preserving (CPTP) maps, also known as quantum channels, arise in quantum theory as super-operators, i.e., mapping density operators to density operators rather than state vectors to state vectors. They represent the time evolution, say from time $0$ to time $T$, also in cases in which the time evolution is not unitary, such as cases given by a Lindblad equation, or involving quantum measurements, or open systems in general.

Our motivation for the present paper comes from the discussion of whether quantum theory can be defined in a relativistic background space-time with closed timelike curves (CTCs), specifically from a proposal by David Deutsch \cite{deutsch_quantum_1991} in which a CPTP map is associated with going around a CTC once, and the existence of a fixed point amounts to the consistency of the theory. We discuss Deutsch's proposal and the application of our result to it in Section \ref{section3}.

Quantum channels are also considered in quantum information theory, in particular as models of noisy systems, where they are typically assumed to act on finite dimensional Hilbert spaces. However, also there it may be of interest to consider an infinite-dimensional Hilbert space, either if it is the fundamentally correct one or if it provides a model that is simpler in some other respect (e.g., by obeying symmetries or being a bosonic Fock space).

Our result, formulated precisely as Theorem \ref{thm1} in Section \ref{section2}, provides a sufficient condition for the existence of a fixed point of a CPTP map $S$. The condition is expressed in terms of several commuting positive self-adjoint operators $A_1,\ldots, A_m$; for these, we have in mind that states with higher values of these observables are more ``costly'' 
to prepare. For example, the $A_i$ could be particle number or kinetic energy, as more resources (``cost'') are required to prepare states with higher particle number or energy. We consider a set $K$ of density operators for which, intuitively speaking, the average cost is limited; more precisely, we consider an axiparallel hyper-rectangle $B$ in $\mathbb{R}^m$ (which intuitively speaking contains the allowed, limited, average values of
$A_1,\ldots,A_m$) and define $K$ to be the set of density operators for which the averages of $A_1,\ldots,A_m$ lie in $B$. Our condition is that $S$ maps $K$ into itself, $S(K)\subseteq K$. (That is, that there are $A_1,\ldots,A_m$ and $B$ such that $S(K)\subseteq K$.) This invariance could be intuitively understood as roughly saying that $S$ will not yield a costly output from a non-costly input.

The proof of our theorem is based on the Schauder-Tychonoff fixed point theorem, which states that for a nonempty, convex, and compact set $K$ in a locally convex space $X$ and a continuous mapping $S:K\to K$, there exists a fixed point $p\in K$, i.e., $S(p)=p$ (Theorem 5.28 in \cite{rudin_functional_1991}).

The remainder of this paper is organized as follows. In Section \ref{section2}, we formulate our main result, Theorem \ref{thm1}. In Section \ref{section3}, we describe Deutsch's proposal for defining a quantum theory on a curved space-time with CTCs and explain what our theorem entails for it. In Section \ref{Proof}, we prove Theorem \ref{thm1}.

\section{Main result}
\label{section2}
We begin with some definitions to recap the concepts of being completely positive and trace preserving: Let $\Hilbert$ be a separable Hilbert space (i.e., one that possesses a countably infinite or finite orthonormal basis) and let $T(\Hilbert)$ be the trace class of $\Hilbert$ (the set of operators $\rho:\Hilbert\to\Hilbert$ with $\Tr(\sqrt{\rho \rho^\dag})<\infty$). 
By a \emph{super-operator}, we mean a linear mapping $S:T(\Hilbert)\to T(\Hilbert)$ that is bounded in the trace norm 
\begin{equation}
\label{trace norm definition}
    \norm{\rho}_1 \coloneqq \Tr(\sqrt{\rho \rho^\dag}).
\end{equation}
$S$ is called \emph{positive} if for every positive $\rho\in T(\Hilbert)$, $S(\rho)$ is also positive. $S$ is called \emph{completely positive} if for every $D\in \mathbb{N}$ and every positive $\eta\in T(\Hilbert\otimes \mathbb{C}^D)$, $(S\otimes Id)(\eta)$ is also positive \cite{breuer_theory_2009}. $S$ is called \emph{trace-preserving} if $\Tr(S(\rho))=\Tr(\rho)$ for every $\rho\in T(\Hilbert)$. A completely positive, trace-preserving super-operator is also called a \emph{CPTP map} or a \emph{quantum channel}.

We denote the set of density operators, i.e., the set of all positive self-adjoint linear operators of trace 1 on a Hilbert space $\Hilbert$ by $D(\Hilbert)$. Note that every quantum channel maps $D(\Hilbert)$ to itself.

Two unbounded self-adjoint operators $A_1$ and $A_2$ are said to \textit{commute} if all spectral projections of the two operators commute. A set of unbounded self-adjoint operators is said to commute if they all commute pairwise.

\begin{thm}
 \label{thm1}
    Let $\Hilbert$ be a separable Hilbert space and $S: T(\Hilbert) \to T(\Hilbert)$ a quantum channel.
     Let $A_1, \hdots, A_m$ be commuting, positive possibly unbounded self-adjoint operators on $\Hilbert$. 
     Let $B \subset \mathbb{R}^m$ be a hyper-rectangle of the form $B\coloneqq [0,b_1]\times \hdots \times [0,b_m]$ for positive numbers $b_1,\hdots, b_m \in (0,\infty)$  and set
     \begin{equation}
     \label{Klabel}
         K \coloneqq \{\rho \in D(\Hilbert)|\: \rho A_i \in T(\Hilbert) \: \forall i, \, (\Tr(\rho A_1), \hdots, \Tr(\rho A_m))\in B
         \}.
     \end{equation}
     Suppose that $S(K) \subseteq K$.
     Suppose also that additionally for all choices of $m$ non-negative numbers $a_1, \hdots , a_m \in \mathbb{R}$ with $0\leq a_i <\infty $  $\forall i\in \{1,\hdots,m\}$, 
     \begin{equation}
     \label{secondCondition}
         \dim(\Hilbert_1 \cap \hdots \cap \Hilbert_m)<\infty,
     \end{equation}
     where $\Hilbert_i$ denotes the spectral subspace of $A_i$ for the interval $[0,a_i]$.\\
     Then either $K=\emptyset$ or there exists a density operator $\rho \in K \: \text{such that } \: S(\rho)=\rho$.
\end{thm}

\begin{remark}
\label{PVMintro}
    The conditions \eqref{Klabel} and \eqref{secondCondition} of Theorem \ref{thm1} can be equivalently rephrased in terms of a projection valued measure (PVM) $P(\omega)$ on the Borel sigma algebra of $\mathbb{R}^m$:\\
    From the PVM we can define the commuting positive self-adjoint operators $A_i$ as \begin{equation} \label{Aiswerdendefiniert}
        A_i=\int_{\mathbb{R}_+^m} P(d\Vec{x}) x_i
    \end{equation}
    where $\mathbb{R}_+=[0,\infty)$.
    Conversely, by the spectral theorem we know of the unique existence of a PVM jointly diagonalizing the given $A_i$ (\cite{reed_methods_1980}, Chapter VIII.3). Using these relations, Theorem \ref{thm1} can equivalently be stated as:\\
    \textit{
    Let $\Hilbert$ be a separable Hilbert space, $S: T(\Hilbert) \to T(\Hilbert)$ a quantum channel, and $P(\omega)$ a PVM on the Borel sigma algebra of $\mathbb{R}^m$ supported only on the positive octant $\mathbb{R}_+^m$.
    Let $B \subset \mathbb{R}^m$ be a hyper-rectangle of the form $B\coloneqq [0,b_1]\times \hdots \times [0,b_m]$ for positive numbers $b_1,\hdots, b_m \in (0,\infty)$  and set
    \begin{equation}
    \label{KdefinitionPVM}
        K \coloneqq \Biggl\{\rho \in D(\Hilbert) \Bigg|  \int_{\mathbb{R}_+^m} \Tr(\rho P(d\Vec{x}))\Vec{x} \in B
         \Biggr\}.
    \end{equation}\\
    Suppose that $S(K) \subseteq K$. 
    Suppose also that additionally for all choices of $m$ non-negative numbers $a_1, \hdots , a_m \in \mathbb{R}$ with $0\leq a_i <\infty $  $\forall i\in \{1,\hdots,m\}$,
    \begin{equation}
        \Tr \left(P([0, a_1]\times \hdots \times [0, a_m])\right) <\infty.
    \end{equation}
    Then either $K=\emptyset$ or there exists a density operator $\rho \in K \: \text{such that } \: S(\rho)=\rho$.}
\end{remark}

 \begin{remark}
     Note that $\rho P(M) \in T(\Hilbert)$ for every set $M$ because $P(M)$ is bounded and the product of a trace-class operator and a bounded operator always lies in the trace class (\cite{reed_methods_1980}, Thm. IV.19(b)). Also note and that $\int_{\mathbb{R}^m_+} \Tr(\rho P(d\Vec{x}))x_i < \infty$ implies $\rho A_i \in T(\Hilbert)$. 
     \footnote{\label{Footnote1} Proof: By definition, $\rho A_i$ lies in the trace class if $\Tr |\rho A_i| < \infty$. Since $\rho$ and $A$ are self-adjoint, $|\rho A_i| = \sqrt{\rho A_i A_i \rho}= \rho^{1/2} |A_i| \rho^{1/2}$; since $A_i$ is positive, $|A_i|=A_i$ and $|\rho A_i| = \int_{\mathbb{R}_+^m} \rho^{1/2} P(d\vec{x}) \rho^{1/2} x_i$. Since $\rho^{1/2} P(M) \rho^{1/2}$ is positive, the Fubini-Tonelli theorem yields that trace and integration can be interchanged, so $\Tr |\rho A_i|= \Tr \int_{\mathbb{R}_+^m}\rho^{1/2} P(d\vec{x}) \rho^{1/2} x_i= \int_{\mathbb{R}_+^m}\Tr (\rho^{1/2} P(d\vec{x}) \rho^{1/2}) x_i= \int_{\mathbb{R}_+^m}\Tr (\rho P(d\vec{x})) x_i$, where the last step follows if the trace is evaluated in an orthonormal eigenbasis of $\rho$.}
 \end{remark}

\begin{remark}
    Theorem \ref{thm1} can be slightly generalized with respect to the allowed set $B \subset \mathbb{R}^m$: It is also valid for a compact convex $B\subset \mathbb{R}^m_+$ with the additional requirement that for all $(x_1, \hdots, x_m) \in B$ the elements $(0, x_2 \hdots, x_m),\, (x_1, 0, x_3 \hdots, x_m), \hdots,\, (x_1, \hdots, x_{m-1}, 0)$ are also elements of $B$.
\end{remark}

\begin{remark}
    \textit{Let $\Hilbert$ be a finite dimensional Hilbert space and $S: T(\Hilbert) \to T(\Hilbert)$ a quantum channel.
    Then there exists a fixed point $\rho \in D(\Hilbert)$ of $S$, i.e., $S(\rho)=\rho$.} \\
    That is, for a finite dimensional Hilbert space, the existence of a fixed point of any quantum channel follows even without the additional restrictions. This can be deduced from elementary considerations \cite{deutsch_quantum_1991} which we include here for the reader's convenience: \\
    Notice first that set of all density operators on a finite-dimensional Hilbert space is compact (as it is a closed bounded set in the finite-dimensional space $T(\Hilbert)$). Consider now the expression
    \begin{equation}
    \label{converges}
         {\rho}(N)=\frac{1}{N+1}\sum^N_{n=0} S^n ( {\rho}(0)),
    \end{equation}
    where $ {\rho}(0)$ is any density operator; $ {\rho}(N)$ is the mean of $N+1$ density operators and therefore itself a density operator. 
    Since
    \begin{equation}
        S(\rho(N))-\rho(N)=\frac{1}{N+1} (S^{N+1}(\rho(0))-\rho(0)),
    \end{equation}
    because the left hand side is a telescopic sum of the right hand side,
    the trace norm of $S({\rho}(N))- {\rho}(N)$ obeys
    \begin{equation}
    \label{Deutschproof}
        0 \leq \norm{S( {\rho}(N))- {\rho}(N)}_1  \leq \frac{2}{N+1}  .
    \end{equation}
    The theorem of Bolzano-Weierstraß states that any bounded sequence in a finite-dimensional vector space has at least one accumulation point \cite{konigsberger_analysis_2004}. 
    Thus, the sequence $ {\rho}(N)$, which lies in the bounded set $D(\Hilbert)$, must have an accumulation point $ {\rho}_{ac}$, so $\rho_{ac}=\lim_{k\to\infty} \rho(N_k) $ for some subsequence $N_k \to \infty$.
    Applying \eqref{Deutschproof} to $N_k$, we find that $S(\rho(N_k))\to \rho_{ac} $ as $k\to \infty$. By continuity of $S$,
    \begin{equation}
        S(\rho_{ac})=S(\lim_{k\to\infty} \rho(N_k))=\lim_{k\to\infty} S(\rho(N_k))=\rho_{ac},
    \end{equation}
    so $\rho_{ac}$ is a fixed point of $S$, and the claim is proved.\footnote{Using Brouwer's fixed-point theorem, an alternative proof is possible \cite{wolf_quantum_2012}. Brouwer's fixed-point theorem states that any continuous function $f$ mapping a nonempty compact convex set in a finite-dimensional vector space to itself, there is a point $x_0$ in that set such that $f(x_0)=x_0$ \cite{ben-el-mechaieh_elementary_2022}.
    The quantum channel $S$ is a linear operator mapping the set $D(\Hilbert)$ to itself.
    Because $S$ is linear, it is also continuous. The set of all density operators acting on $\Hilbert$ is non-empty, compact and also convex.
    Therefore, there exists at least one fixed-point ${\rho}_0 \in D(\Hilbert)$ such that $S({\rho}_0)={\rho}_0$. }\\
    The sequence $\rho(N)$ in \eqref{converges} does not only have a convergent subsequence, but actually converges for any initial $\rho(0) \in D(\Hilbert)$ (Chapter 6 in \cite{wolf_quantum_2012}).
\end{remark}

\begin{remark}
    \textit{In infinite dimensions a quantum channel can fail to have a fixed point.}\\
      Here is an example: let $\Hilbert=\ell^2 = L^2(\mathbb{N})$ be the space of square-summable sequences and $V$ the right shift, 
      \begin{equation}
          V(x_1,x_2,\hdots)=V(0,x_1,x_2,\hdots).
      \end{equation}
      Then the map $S(\rho)=V\rho V^\dag$ is a quantum channel without a fixed point.
 \end{remark}   
 \begin{remark}
     \textit{Let $P$ be as in Eq. \eqref{Aiswerdendefiniert} and $B\subset \mathbb{R}^m$. If $P(B)\neq 0$, then every}
     \hbox{$\rho \in D(\text{range}[P(B)])$} \textit{lies in $K$; in particular $K \neq \emptyset$. }\\
     Indeed, we can take \eqref{KdefinitionPVM} as the definition of $K$. Since $\rho P(B)=\rho$, we have that $\rho P(C)=0$ for any set $C\subset \mathbb{R}^m$ disjoint from $B$. 
     Therefore, the probability measure $\Tr(\rho P(d\Vec{x}))$ on $\mathbb{R}^m$ is concentrated in the set $B$. 
     Since $B$ is bounded, this measure has finite expectation value $\int_{\mathbb{R}_+^m} \Tr(\rho P(d\vec{x})) \vec{x}$; since $B$ is convex, the expectation value lies in $B$. 
     Thus, any $\rho$ concentrated in range$(P(B))$ lies in $K$. 
     Since $P(B)\neq 0$, range$(P(B))$ has nonzero dimension, so there exist such $\rho$.
 \end{remark}

\section{Application to retrocausality}
\label{section3}
In this section we explain the relevance of fixed point theorems, in particular Theorem \ref{thm1}, to quantum mechanics on space-times with causality violations, in particular with closed timelike curves (CTCs). 
While we do not claim that CTCs exist in reality, we are interested in the theoretical challenge of whether and how quantum mechanics can be formulated on a space-time with CTCs. 
As we explain below, the strategy proposed for this purpose by Deutsch \cite{deutsch_quantum_1991} seems reasonable and viable, even more so in the light of Theorem \ref{thm1}.

\subsection{Closed Timelike Curves (CTCs)}

The framework of general relativity, in which space-time is a Lorentzian manifold, in principle allows for CTCs. Explicit metrics with CTCs have come up in various contexts \cite{godel_example_1949, thorne_closed_1993, bonnor_exact_2005, carter_Kerr}, and there is considerable literature on conditions under which the Einstein field equation of general relativity would lead to the formation of CTCs, or conditions that would preclude such solutions (e.g., \cite{friedman_cauchy_1990, kim_vacuum_1991}), or whether those CTCs would be traversable. Other kinds of retrocausality have also been studied, for example in classical electrodynamics in flat space-time by Wheeler and Feynman \cite{wheeler_classical_1949}.

The obvious problem arising from space-times with CTCs is that if a particle were to traverse a CTC, it would encounter its own past. The following causality paradox is often called the grandfather paradox and has a long
history: A person has figured out how to time travel. The person goes back in time and kills his grandfather. Therefore the time traveler is never born---a contradiction. We can also formulate an analogous paradox for a single particle: Consider a scenario where a particle, following a CTC, collides with its younger self, altering its trajectory so that its younger self now follows a different path. Following that path, however, the particle will later fail to collide with its younger self, yielding a contradiction.

It thus seems that the existence and traversability of CTCs leads to logical inconsistencies and thus precludes the possibility of physical theories, classical or quantum, on such a space-time. Some authors \cite{arntzenius_time_2002} have argued that the laws of physics will still have sufficiently many consistent solutions to allow for a plausible course of events avoiding any contradiction. We now turn to a particular way out of the paradox for quantum theories proposed by Deutsch \cite{deutsch_quantum_1991}.

\subsection{Model of quantum retrocausality}

The simplest toy example $\mathscr{M}_1$ of a space-time with CTCs is obtained from Minkowski space-time by ``rolling it up'': consider the slab between two parallel spacelike 3-planes $\mathcal{P}_1,\mathcal{P}_2$, and identify $\mathcal{P}_1$ with $\mathcal{P}_2$ along a time translation $\tau$ that maps $\mathcal{P}_1$ to $\mathcal{P}_2$; that is, a future-directed timelike current hitting the past side of $\mathcal{P}_2$ at $x$ comes out from the future side of $\mathcal{P}_1$ at $\tau^{-1}(x)$. In a Lorentz frame in which $\mathcal{P}_1$ and $\mathcal{P}_2$ are surfaces of constant time coordinate and $\tau$ translates along the time axis (see Figure \ref{fig:Sketchoftimecylinder}), $\mathcal{M}_1$ can be written as
\begin{equation}
    \left( \mathbb{R}_{/T\mathbb{Z}} \times \mathbb{R}^3 , ds^2=- dt^2+ dx_1^2+dx_2^2+dx_3^2\right)
\end{equation}
where $T>0$ is the timelike distance between $\mathcal{P}_1$ and $\mathcal{P}_2$. The symbol $\mathbb{R}_{/T\mathbb{Z}}$ means that we are considering the set $[0,T)$ where the upper end is glued to the lower end, i.e., $T$ is identified with zero. Topologically this corresponds to a circle $\mathbb{S}^1$, so the $t$ coordinate can now be regarded as cyclic.
\begin{figure}[ht]
    \centering
    \includegraphics[scale = 1]{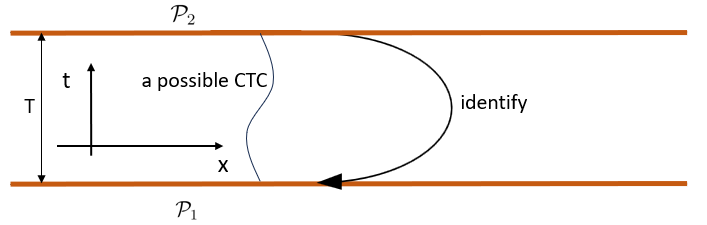}
    \caption{Space-time diagram of the simple causality violating space-time $\mathscr{M}_1$}
    \label{fig:Sketchoftimecylinder}
\end{figure}

Considering wave functions on $\mathscr{M}_1$ leads to the following problem: 
Since $\mathcal{P}_2$ is really the same set of space-time points as $\mathcal{P}_1$, the wave function on $\mathcal{P}_2$ should agree with that on $\mathcal{P}_1$; or in coordinates, $\psi(t=T)=\psi(t=0)$. However, for a unitary time evolution operator $e^{-iHT/\hbar} $ representing the time evolution for time length T, the equation

\begin{equation}
    \ket{\Psi}= e^{-iHT/\hbar} \ket{\Psi}
\end{equation}
does not necessarily have a solution $\ket{\Psi}\in L^2(\mathbb{R}^3)$. It only has a solution if $H$  has an eigenvalue of $ \frac{2\pi \hbar}{T}  k$ for a $k \in \mathbb{Z}$. For generic $H$, this is not the case (and even if such eigenvalues exist, the solutions $\ket{\Psi}$ constrained in this way are presumably too few for an acceptable physical theory).

For density operators the situation is different. Any $\rho$ diagonal in an eigenbasis of the Hamiltonian $H$ governing the time evolution is invariant under the time evolution. Therefore it is the same after a circumnavigation of the time cylinder described above. Having such $\rho$ is therefore logically consistent with no constraints on the Hamiltonian or the length of the time cylinder.

Thus, the crucial problem of the consistency of the theory in the presence of closed timelike curves can be, it seems, much improved if we are willing to contemplate the possibility that the quantum state of the universe is fundamentally mixed---not because of the observers' ignorance of the actual wave function, and not either because of tracing out some degrees of freedom, but directly on the fundamental level. 
This possibility, that the fundamental quantum state may not be a wave function but a density operator, has been considered before for other reasons under the name ``density matrix realism'' \cite{durr_role_2005, Chen25}.

\subsection{Another model}
To obtain a more nuanced model of CTCs allowing for a causality violating region as well as a chronology respecting region we modify a Minkowski space-time to obtain a space time $\mathscr{M}_2$ in the following way:
We choose a space-volume $\mathcal{F}\subset \mathbb{R}^3$
and two times $T_1$ and $T_2$ with $T_1<T_2$ to obtain two spacelike sets
\begin{equation}
    \mathcal{F}_i\coloneqq \{T_i\} \times \mathcal{F}
\end{equation}
We then identify $\mathcal{F}_1$ with $\mathcal{F}_2$. A sketch is given in Figure \ref{fig:SketchofHilbertSpaces}.
This model can incorporate CTCs because it transports whole states from one space-time surface at time $T_2$ to the other at the earlier time $T_1$.
\begin{figure}[ht]
    \centering
    \includegraphics[scale = 0.7]{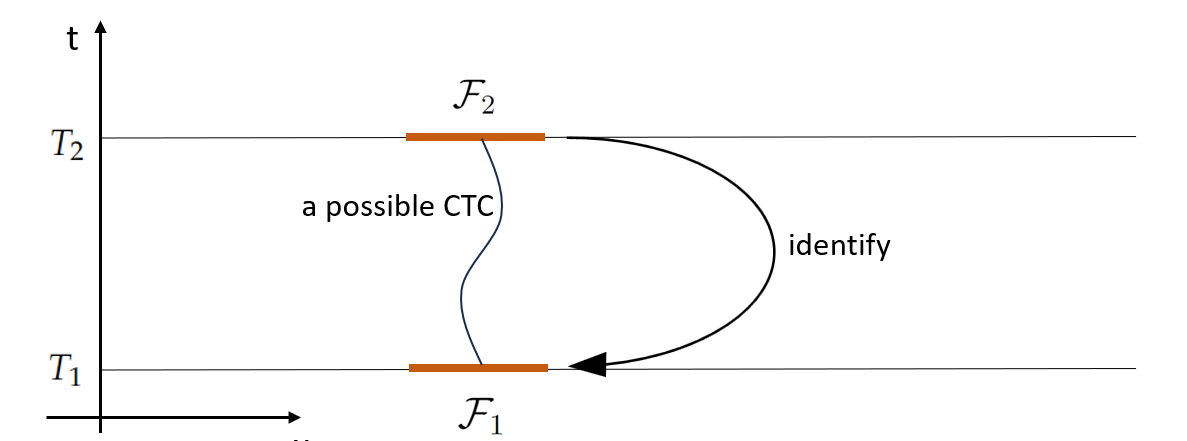}
    \caption{Space-time diagram of Deutsch's model space-time $\mathscr{M}_2$ with the chronology respecting region to the past of $T_1$, the chronology respecting region to the future of $T_2$ and the region with CTCs between $T_1$ and $T_2$}
    \label{fig:SketchofHilbertSpaces}
\end{figure}
\\
The normal unitary evolution applies in the chronology respecting regions prior to $T_1$ and after $T_2$.\\
In the causality violating region, we define the time evolution as 
follows: We assume that for regions $\mathcal{R} \subset \mathbb{R}^3$, 
the Hilbert space factorizes according to
\begin{equation}\label{factorize}
\mathscr{H}_{\mathrm{full}} = \mathscr{H}_{\mathcal{R}^c} \otimes 
\mathscr{H}_{\mathcal{R}} \,,
\end{equation}
where $\mathcal{R}^c = \mathbb{R}^3 \setminus \mathcal{R}$ is the 
complement of the region $\mathcal{R}$. Eq.\ \eqref{factorize} is known 
to be true for bosonic and fermionic Fock spaces; we will apply it to 
$\mathcal{R}=\mathcal{F}$. We postulate first that the density operator 
$\rho_{T_1+}$ (right after time $T_1$) on $\mathscr{H}_{\mathrm{full}}$ 
arises from a given $\rho_{T_1-}$ (right before $T_1$) by combining a 
state $\rho_1$ coming out of the future side of $\mathcal{F}_1$ with the 
prior reduced state, 
\begin{equation}
\rho_{in}:=\Tr_{\mathcal{F}}(\rho_{T_1-})
\end{equation}
on $\mathcal{F}^c$, which we call the initial state from the chronology respecting region. The combination of the two is the following:
\begin{equation}
\rho_{T_1+} =  \rho_{in} \otimes \rho_1 \,.
\end{equation}
We postulate next that $\rho$ evolves unitarily between $T_1+$ and $T_2-$,
\begin{equation}
\rho_{T_2-}= U \, \rho_{T_1+} \, U^\dagger
\end{equation}
with $U=e^{-iH(T_2-T_1)/\hbar}$. Finally, we postulate that the reduced 
state
\begin{equation}
\rho_2:=\Tr_{\mathcal{F}^c}(\rho_{T_2-})
\end{equation}
on $\mathcal{F}$ gets transported back from $\mathcal{F}_2$ to 
$\mathcal{F}_1$.\footnote{For the region after $T_2$, we may postulate, 
e.g., $\rho_{T_2+}= \Tr_{\mathcal{F}}(\rho_{T_2-}) \otimes |0\rangle 
\langle 0|$ (with $|0\rangle$ the Fock vacuum vector) or $\rho_{T_2+}= 
\Tr_{\mathcal{F}}(\rho_{T_2-}) \otimes 
\Tr_{\mathcal{F}^c}(\rho_{T_1-})$, but that does not matter for the 
considerations in this paper.} Consistency of the time evolution then 
amounts to
\begin{equation}
\rho_2 = \rho_1
\end{equation}
or
\begin{equation}
S(\rho_1)=\rho_1
\end{equation}
with $S$ given by
\begin{equation}
\label{quantumchannel}
S(\rho) = \Tr_{in} (U \, \rho_{in} \otimes \rho \, U^\dagger) \,,
\end{equation}
where $\Tr_{in}$ means the same as $\Tr_{\mathcal{F}^c}$, i.e., a partial trace over the initial degrees of freedom from the chronology respecting region. For given 
$\rho_{in}$, the mapping $S$ represents the time evolution from 
$\mathcal{F}_1$ to $\mathcal{F}_2$ and back, and thus going around the 
CTC once.
The consistency of the postulated dynamics is equivalent to the existence of fixed points of $S$. As mentioned before, for the case of finite dimensional quantum mechanics, such a point exists for every $\rho_{in}\in \Hilbert_{\mathcal{F}^c}$.

\begin{remark}
\label{S is a quantum channel}
    \textit{Let $\Hilbert_{\mathcal{F}^c}$ and $\Hilbert_{\mathcal{F}}$ both be separable Hilbert spaces. Then the map \eqref{quantumchannel} is a quantum channel.}\\
    Indeed, $S$ is trace preserving: Let $\rho \in T(\Hilbert)$ be a trace class operator on $\Hilbert$, then
    \begin{equation}
    \begin{split}
        \Tr(S(\rho))&=\Tr(\Tr_{in}(U{\rho}_{in}\otimes {\rho}U^\dag))\\
        &=\Tr(U{\rho}_{in}\otimes {\rho}U^\dag)\\
        &=\Tr({\rho}_{in}\otimes {\rho})= \Tr(\rho_{in})\Tr(\rho)=\Tr(\rho)
    \end{split}
    \end{equation}
    as $\rho_{in}$ has trace 1.
    
    $S$ is completely positive: Let $\rho\in T(\Hilbert \otimes \mathbb{C}^N)$ be a positive operator for a $N\in \mathbb{N}$. 
    Let $\{e_n\}$ be an orthonormal basis of $\Hilbert_{in}$ which is being traced out. Then for a $\ket{\psi} \in \Hilbert \otimes \mathbb{C}^N$ we have:
    \begin{equation}
    \begin{split}
        \bra{ \psi } (S \otimes \text{Id})(\rho) \ket{ \psi } 
    &= \bra{ \psi } \Tr_{in}\left( (U \otimes \text{Id})(\rho_1 \otimes \rho)(U^\dag \otimes \text{Id}) \right) \ket{\psi}\\
    & = \sum_n \bra{e_n} \bra{\psi} (U \otimes \text{Id})(\rho_1 \otimes \rho)(U^\dag \otimes \text{Id})  \ket{\psi} \ket{e_n}\geq 0
    \end{split}
    \end{equation}
    because $T := \rho_1 \otimes \rho$ is a positive operator, and $R T R^\dag$ is always positive for a unitary $R$.
    
    $S$ is bounded in the trace norm \eqref{trace norm definition}, which is equivalent to $S$ being continuous in the trace norm:
    To prove boundedness, we will consider all three operations of $S$ separately.
    \begin{itemize}
        \item ${\rho} \mapsto {\rho}_{in} \otimes {\rho}$ is bounded, because \begin{equation}
            ||{\rho}_{in} \otimes {\rho}||_1 = ||{\rho}_{in}||_1  ||{\rho}||_1 .
        \end{equation} $||{\rho}_{in}||_1$ is finite because ${\rho}_{in}$ is a trace class operator with trace norm 1. 
        \item ${\rho} \mapsto U{\rho}U^\dag$ is bounded, because \begin{equation}
            ||U{\rho}U^\dag||_1 \leq ||U||_{Op}  ||{\rho}U^\dag||_1 \leq ||U||_{Op}  ||{\rho}||_1||U^\dag||_{Op}=||{\rho}||_1.
        \end{equation} $||*||_{Op}$ represents the operator norm. The norm inequality can be derived from Hölder's inequality $\norm{AB}_1\leq \norm{A}_1 \norm{B}_{Op}$, which holds for the trace norm \cite{noauthor_schatten_2023}.
        \item ${\rho} \mapsto \Tr_{in}({\rho})$ is bounded because the inequality $||\Tr_{in}({\rho})||_1 \leq ||{\rho}||_1$ holds for the trace norm \cite{rastegin_relations_2012}.
    \end{itemize}
    The composition of bounded maps is again bounded. Therefore the map $S$ is bounded in the trace norm. It is therefore a quantum channel. 
\end{remark}

\subsection{The infinite-dimensional case}

As mentioned, the consistency of Deutsch's proposed dynamics amounts to the existence of fixed points of $S$, which is guaranteed if the Hilbert space $\mathscr{H}_{\mathcal{F}}$ has finite dimension but not otherwise. We now explain how our Theorem \ref{thm1} can be applied to the infinite-dimensional case and provides the existence of a fixed point under physically plausible conditions. Essentially, these conditions say that the average particle number and energy cannot be increased without limits through iterations of $S$. They seem plausible if we imagine that even for a time machine it would be too good to be true if we could arbitrarily multiply resources (such as particles and energy) by letting them repeatedly traverse a CTC.\footnote{For example, it seems plausible that high pressure in the rectangle $[T_1,T_2]\times \mathcal{F}$ between $\mathcal{F}_1$ and $\mathcal{F}_2$ in Figure \ref{fig:SketchofHilbertSpaces} leads to the expulsion of particles from this rectangle. Furthermore, it would seem that the space-time is unstable if energy can accumulate and increase from seemingly nowhere; after all, by the Einstein equation this large amount of energy would cause a change of the space-time geometry, perhaps even removing the CTCs. Finally, our condition can be regarded as a weak form of conservation of particle number and energy. Friedman et al. (\cite{friedman_cauchy_1990}, Sec. II.F) argued that conservation laws, suitably understood, should remain valid in space-times with CTCs.} Here is a precise statement of our assumptions and the corollary obtained with their help from Theorem \ref{thm1}.

\begin{itemize}
\item[(A1)] We are given a separable Hilbert space $\mathscr{H}_{1\mathcal{F}}$ representing the 1-particle Hilbert space associated with the space volume $\mathcal{F}$. We are also given a 1-particle Hamiltonian $H_1$ in $\mathscr{H}_{1\mathcal{F}}$ that is positive, has purely discrete spectrum, and has only finitely many eigenvalues including multiplicity up to any chosen energy $e>0$.
\end{itemize}

It is a common assumption that 1-particle Hamiltonians are positive. If $\mathcal{F}$ is a bounded region, then that usually forces the spectrum of differential operators (such as the 1-particle Laplacian or Dirac
operator) to be discrete, with only finitely many eigenvalues in every energy interval.

\begin{itemize}
\item[(A2)] We take $\mathscr{H}_{\mathcal{F}}$ to be the bosonic or fermionic Fock space over $\mathscr{H}_{1\mathcal{F}}$, \begin{equation} \mathscr{H}_{\mathcal{F}}= \bigoplus_{n=0}^\infty \mathcal{S}_{\pm} \mathscr{H}_{1\mathcal{F}}^{\otimes n} \end{equation} with $\mathcal{S}_{\pm}$ the (anti-)symmetrization operator.

\item[(A3)] We take $m=2$, $A_1=\hat{N}$ as the particle number operator on $\mathscr{H}_{\mathcal{F}}$, and $A_2=\hat{E}$ as the free (i.e.,
non-interacting) Hamiltonian obtained via second quantization of $H_1$.
\end{itemize}

It then follows using (A1) that $A_1$ and $A_2$ are commuting positive operators and that the joint spectral subspace of any bounded set has finite dimension (see the proof of Corollary \ref{Corollary}).

\begin{Corollary}
\label{Corollary}
    \textit{Assume (A1)--(A3). Let $\mathscr{H}_{in}$ be any separable Hilbert space, $\rho_{in}\in D(\mathscr{H}_{in})$, $U$ any unitary on $\mathscr{H}_\mathrm{full}:= \mathscr{H}_{in} \otimes \mathscr{H}_{\mathcal{F}}$, and let $S$ be given by \eqref{quantumchannel}. 
    Suppose there exist real numbers $N,E>0$ such that \begin{equation} K=\Bigl \{ \rho \in D(\mathscr{H}_{\mathcal{F}}) \Big| \rho\hat{N}, \rho\hat{E} \in T(\mathscr{H}_{\mathcal{F}}) \text{ and } \Tr(\rho
\hat{N}) \leq N \text{ and } \Tr(\rho \hat{E})\leq E \Bigr\} \end{equation} gets mapped by $S$ into itself. Then $S$ possesses a fixed point in the set $K$.}
   
\end{Corollary}

\begin{proof}
    The Hilbert space $\mathscr{H}_\mathrm{full}$ is separable by construction, and $S$ is a quantum channel according to Remark \ref{S is a quantum channel} above. 
    The operators $A_1=\hat{N}$ and $A_2=\hat{E}$ are clearly self-adjoint; they are both positive because they are second quantizations of positive operators, $\hat{N}=d\Gamma(I)$ and $\hat{E}=d\Gamma(H_1)$; they commute because $I$ and $H_1$ commute. In our case, the set $B$ is $B=[0,N]\times [0,E]$. 
    Next, by (A1) there exists an orthonormal eigenbasis $(\phi_i)_{i\in\mathbb{N}}$ of $H_1=\sum_i E_i |\phi_i\rangle \langle \phi_i|$ with eigenvalues $0\leq E_1\leq E_2\leq \ldots$; let $i(e)$ be the number with multiplicity of eigenvalues $\leq e$, which is finite by (A1). 
    Then 
    \begin{align}
         \mathrm{rank}\, P([0,a_1]\times [0,a_2]) &\leq \dim \mathrm{span}\bigcup_{n\in \mathbb{Z}\cap [0,a_1]} \Bigl\{ \mathcal{S}_{\pm} (\phi_{i_1}\otimes \cdots \otimes \phi_{i_n}) \Big| i_1,\ldots,i_n\leq i(a_2) \Bigr\} \nonumber\\
         &\leq \sum_{n\in \mathbb{Z}\cap [0,a_1]} i(a_2)^n < \infty.
    \end{align}
Thus, all hypotheses of Theorem \ref{thm1} are satisfied.
The set $K$ is non-empty because $\rho=|0\rangle \langle 0|$ with $|0\rangle$ the Fock vacuum satisfies $\Tr(\rho \hat{N})=0<N$ and $\Tr(\rho \hat{E})=0<E$. Therefore, $S$ has a fixed point in $K$.
\end{proof}
\noindent 
\begin{remark}
    There might be more than one fixed point for a given $\rho_{in}$, and thus more than one consistent solution of Deutsch's time evolution for the density operator in the space-time of Figure \ref{fig:SketchofHilbertSpaces} with the same initial data. This non-uniqueness was already discussed by Deutsch \cite{deutsch_quantum_1991}.
\end{remark}

\section{Proof of Theorem 1}
\label{Proof}
\begin{proof}
We want to utilize the Schauder-Tychonoff fixed-point theorem to prove Theorem \ref{thm1}. Therefore, we will check its hypotheses in the following.
To begin with, we explain why the trace class $T(\Hilbert)$ with the trace norm is a locally convex topological vector space. This holds because by definition of the trace class, each element has finite trace norm, and the trace class together with the trace norm forms a Banach space (\cite{reed_methods_1980} Theorem VI.20). This implies the that the trace class with the trace norm is a locally convex topological vector space.\\
\\
\underline{Part 1: Characteristics of the map $S$}\\ By hypothesis, $S$ maps the set $K$ as defined in \eqref{Klabel} into $K$. 
Since $S$ is a quantum channel, it is by definition bounded in the trace norm and therefore continuous in the trace norm.\\
\\
\underline{Part 2: $K$ is convex}\\
Consider two arbitrary elements ${\rho}_a$ and ${\rho}_b$ of $K$.
We need to show that the set
\begin{equation}
    G=\{  \alpha {\rho}_a + (1-\alpha) {\rho}_b | \alpha \in [0,1] \}
\end{equation} is a subset of $K$.
For every ${\rho} \in G$, ${\rho}\in D(\Hilbert)$ and 
\begin{equation}
    \begin{split}
         \begin{pmatrix}
           \Tr({\rho} A_1) \\
           \vdots \\
           \Tr({\rho} A_m)
         \end{pmatrix} &=
         \begin{pmatrix}
           \Tr([\alpha {\rho}_a + (1-\alpha) {\rho}_b] A_1) \\
           \vdots \\
           \Tr([\alpha {\rho}_a + (1-\alpha) {\rho}_b] A_m)
         \end{pmatrix} \\
         &=
         \begin{pmatrix}
           \alpha \Tr( {\rho}_a A_1)+ (1-\alpha) \Tr({\rho}_b A_1) \\
           \vdots \\
           \alpha \Tr( {\rho}_a A_m)+ (1-\alpha) \Tr({\rho}_b A_m)
         \end{pmatrix} \\&=
         \alpha\begin{pmatrix}
            \Tr( {\rho}_a A_1) \\
           \vdots \\
            \Tr( {\rho}_a A_m)
         \end{pmatrix} + (1-\alpha) \begin{pmatrix}
           \Tr({\rho}_b A_1) \\
           \vdots \\
            \Tr({\rho}_b A_m)
         \end{pmatrix}
         \\&\eqqcolon (\alpha E_a + (1-\alpha) E_b) \in B
         \end{split}
\end{equation}
because $\rho_a \in K$ implies $E_a \in B$ (analogously for $\rho_b$) and $B$ is convex. This in turn implies that $\rho \in K$.
Therefore, $G$ is a subset of $K$ and the set $K$ is convex.\\
\\
\underline{Part 3: $K$ is compact}\\
In complete metric spaces, a set is compact if and only if it is closed and totally bounded (\cite{munkres_topology_2014} Theorem 45.1). Totally bounded means for a subset $K$ of a metric space that for all $\epsilon >0$, there is a finite collection of open balls of radius $\epsilon$ whose union contains $K$.
If for a subset $K$ of a normed vector space $T$, for all $\epsilon >0$ there is a finite-dimensional subspace $F_\epsilon \subset T$ and a bounded set $L_\epsilon \subset F_\epsilon$ such that any element of $K$ is at most $\epsilon$ away from an element of $L_\epsilon$, then $K$ is totally bounded for the following reason:

For a given $\epsilon>0$ we want to show that $K$ can be covered by finitely many open balls of radius $\epsilon$. Now take the finite dimensional subspace and the bounded set therein for $\frac{\epsilon}{2}$.
We know that $L_{\frac{\epsilon}{2}}$ can be covered by finitely many open balls of radius $\frac{\epsilon}{2}$ because it is a bounded set in a finite-dimensional vector space (so its closure is compact). Denote the center points of these finitely many open $\frac{\epsilon}{2}$-balls by $x_1,\hdots,x_r$. We show that the collection of $\epsilon$-balls around $x_1,\ldots,x_r$ covers $K$. 
Indeed, for every $x\in K$ there exists a $y\in L_{\frac{\epsilon}{2}}$ such that $\norm{x-y} \leq \frac{\epsilon}{2}$ by hypothesis. There also exists an $i\in\{1,\ldots,r\}$ such that $\norm{y-x_i} < \frac{\epsilon}{2}$. By the triangle inequality, $\norm{x-x_i} < \epsilon$, which is what we needed to show.

The trace class operators are a complete metric space with respect to the trace norm. We first show that the set $K$ as defined in \eqref{Klabel} is closed by showing that for any sequence $(\rho_n)_{n\in \mathbb{N}} \subseteq K$ that converges in the trace norm, also $\rho\coloneqq \lim_{n\to \infty } \rho_n \in K $: First, $D(\Hilbert)$ can be shown to be closed in the trace norm by elementary consideration. 

Second we will show that $\Tr(\rho A_i)\in [0,b_i]$ for every
$i\in\{1,\ldots,m\}$: We approximate $A_i$ by bounded operators. 
To this end, we use the PVM $P$ from Remark \ref{PVMintro} to define for every $k\in(0,\infty)$ the projection \hbox{$Q_{ik}:=P(\{x_i\leq k\}) = P(\mathbb{R}^{i-1}_+ \times [0,k] \times \mathbb{R}^{m-i}_+)$.} 
It follows that $Q_{ik} A_i Q_{ik} = \int_{\mathbb{R}^m_+} P(d\vec{x}) \, x_i \, 1_{x_i\leq k}$ is a bounded self-adjoint operator; 
in particular, $\rho Q_{ik} A_i Q_{ik}$ is a trace class operator. Fubini's theorem allows us to write \begin{align} \Tr(\rho Q_{ik} A_i Q_{ik}) &= \Tr \int_{\mathbb{R}^m_+} \rho P(d\vec{x}) \, x_i \, 1_{x_i\leq k}\\ &= \int_{\mathbb{R}^m_+} \Tr(\rho P(d\vec{x})) \, x_i \, 1_{x_i\leq k} \, , \end{align} 
and by the monotone convergence theorem 
\begin{align} \Tr(\rho Q_{ik} A_i Q_{ik}) &\stackrel{k \to \infty}{\longrightarrow}\int_{\mathbb{R}^m_+} \Tr(\rho
P(d\vec{x})) \, x_i\\
&=\Tr(\rho^{1/2} A_i \rho^{1/2}) \, .
\end{align}
Now since $0\leq Q_{ik} A_i Q_{ik} \leq A_i$ and $\rho_n \in K$ by hypothesis, we have that \begin{equation} 0\leq \Tr(\rho_n Q_{ik} A_i Q_{ik}) \leq \Tr(\rho_n A_i) \leq b_i \, .
\end{equation}
Therefore, also $\Tr(\rho^{1/2} A_i \rho^{1/2}) \in [0,b_i]$. In particular, $\Tr |\rho A_i| <\infty$ as in Footnote \ref{Footnote1}, so $\rho A_i$ lies in the trace class, and $\Tr(\rho A_i) = \Tr(\rho^{1/2} A_i
\rho^{1/2}) \in [0,b_i]$, which is what we wanted to show.
This implies that $\rho \in K$, completing the proof that $K$ is closed.

Now we show that $K$ is totally bounded: 
For a given $\epsilon>0$ and every $i\in\{1,\ldots,m\}$ define $n_{i,\epsilon} \coloneqq \min\left( n \in \mathbb{N} \left| 4m\frac{b_i}{n} < \epsilon^2 \right. \right)$ and, with $P$ again the unique PVM jointly diagonalizing to the commuting self-adjoint operators $A_1, \hdots, A_m$,
\begin{equation}
    P_\epsilon \coloneqq P\left(\prod_{i=1}^m [ 0, n_{i,\epsilon}] \right),
\end{equation} which is of finite rank by hypothesis.\\
The Markov inequality states that for a non-negative random variable $X$ and $C>0$,
\begin{equation}
    \mathds{P}(X\geq C)\leq \frac{\mathds{E}(X)}{C}.
\end{equation}
Fix $\rho \in K$ and let the random variables $X_1,\hdots, X_m$ have joint probability distribution $\Tr(\rho P(d\Vec{x}))$. Then $\mathds{E}(X_i)=\Tr(\rho A_i)\leq b_i$, and the Markov inequality for $X_i$ and $C\coloneqq n_{i,\epsilon}> 4m \frac{b_i}{\epsilon^2} $ yields that
\begin{equation}
\begin{split}
    \mathds{P}(X_i \geq n_{i,\epsilon})&\leq \frac{\Tr(\rho A_i)}{n_{i,\epsilon}}< \frac{b_i \epsilon^2}{4m b_i}=\frac{\epsilon^2}{4m}.
\end{split}
\end{equation}
Thus,
\begin{equation}
\label{Markov inequality applied}
\begin{split}
     \Tr(\rho P_\epsilon)&= \mathds{P}(\forall i:\: X_i\leq n_{i,\epsilon}) \\
     &=1-\mathds{P}(\exists i:\: X_i> n_{i,\epsilon})\\
     &\geq 1- \sum_{i=1}^m \mathds{P}(X_i > n_{i,\epsilon})\geq 1-\frac{\epsilon^2}{4}
\end{split}
\end{equation}

For the given $\epsilon >0$ we choose $F_\epsilon$ to be the space of operators $\sigma \in T(\Hilbert)$ with 
\begin{equation}
    P_\epsilon \sigma = \sigma P_\epsilon = \sigma.
\end{equation}
Equivalently, $F_\epsilon$ is the space of finite-rank operators $\sigma$ with $\text{range}(\sigma) \subseteq \text{range}(P_\epsilon) $ and $\text{kernel}(P_\epsilon) \subseteq \text{kernel}(\sigma)$.\\
Or, in terms of any orthonormal basis $\{\chi_i\}$ of $\Hilbert$ in which $P_\epsilon$ is diagonal, the matrix elements $\bra{\chi_i}\sigma\ket{\chi_j}$ are non-zero only when $\ket{\chi_i},\ket{\chi_j} \in \text{range}(P_\epsilon)$.\\
We know that $\dim(F_\epsilon)\leq \dim(\text{range}(P_\epsilon))^2 < \infty$. We take $L_\epsilon$ to be the closed unit ball (in the trace norm) in $F_\epsilon$.

We will consider for an arbitrary ${\rho} \in K$ the element of $x\in F_\epsilon$ which matches the ${\rho}$ in all its finitely many non-zero matrix entries relative to $\{\chi_i\}$, that is, $x=P_\epsilon{\rho}P_\epsilon$.

Note that 
\begin{equation}
    \begin{split}
        \norm{x}_1&=\norm{P_\epsilon \rho P_\epsilon}_1 \leq \norm{P_\epsilon}_{Op} \norm{\rho}_1 \norm{P_\epsilon}_{Op} =1
    \end{split}
\end{equation}
by the Hölder inequality, so $x\in L_\epsilon$. It now suffices to show that
\begin{equation}
    \norm{P_\epsilon \rho P_\epsilon-\rho}_1 \leq \epsilon 
\end{equation}
for every $\rho \in K$, and that is what the remainder of this proof is about.

Let $\{\ket{\Psi_k}\}$ be an orthonormal eigenbasis of $\rho \in \Hilbert$ and $p_k$ the eigenvalue of $ \ket{\Psi_k}$, so
\begin{equation}
    \rho=\sum_k p_k \ket{\Psi_k} \bra{\Psi_k}.
\end{equation}
Then
\begin{equation}
    \begin{split}
         ||P_\epsilon{\rho}P_\epsilon-{\rho}||_1 &=\norm{\sum_k p_k (P_\epsilon\ket{\Psi_k}\bra{\Psi_k}P_\epsilon- \ket{\Psi_k}\bra{\Psi_k})}_1 \\
        &\leq \sum_k p_k \bigg|\bigg|{ P_\epsilon\ket{\Psi_k}\bra{\Psi_k}P_\epsilon -\ket{\Psi_k}\bra{\Psi_k}}\bigg|\bigg|_1 \; .
    \end{split}
\end{equation}

\begin{lem}
The trace norm of the self-adjoint operator
\begin{equation}
    O_k \coloneqq P_\epsilon\ket{\Psi_k}\bra{\Psi_k}P_\epsilon -\ket{\Psi_k}\bra{\Psi_k}
\end{equation}
can be bounded by
    \begin{equation}
    \label{inequalityO_k}
        \norm{O_k}_1 \leq 2 \norm{(\mathds{1}-P_\epsilon)\ket{\Psi_k}}.
    \end{equation}
\end{lem}
\begin{proof}
    If $P_\epsilon \ket{\Psi_k}=0$, then $O_k=-\ket{\Psi_k}\bra{\Psi_k}$, $\norm{O_k}_1=1$, and $\norm{(\mathds{1}-P_\epsilon)\ket{\Psi_k}}=1$, so \eqref{inequalityO_k} holds.

    If $(\mathds{1}-P_\epsilon)\ket{\Psi_k}=0$, then $P_\epsilon\ket{\Psi_k}=\ket{\Psi_k}$ and therefore $O_k=0$, so \eqref{inequalityO_k} holds as well. 

    Now we assume that neither $P_\epsilon \ket{\Psi_k}$ nor $(\mathds{1}-P_\epsilon)\ket{\Psi_k}$ vanish. For ease of notation, define $\ket{\Psi_k}\eqqcolon \ket{\Psi}$. We also define 
\begin{equation}
\begin{split}
     \ket{a}\coloneqq \frac{P_\epsilon \ket{\Psi}}{\norm{P_\epsilon \ket{\Psi}}} \quad & \text{and} \quad \ket{b}\coloneqq \frac{(\mathds{1}-P_\epsilon) \ket{\Psi}}{\norm{(\mathds{1}-P_\epsilon) \ket{\Psi}}}\\
    \alpha \coloneqq \norm{P_\epsilon \ket{\Psi}} > 0  \quad & \text{and} \quad \beta \coloneqq \sqrt{1-\alpha^2} > 0.
\end{split}
\end{equation}
Notice that $\ket{a}$ and $\ket{b}$ are orthogonal and both unit vectors. Furthermore $\ket{a} \in \text{range}(P_\epsilon)$ and $\ket{b} \perp \text{range}(P_\epsilon)$. \\
We use them to make the following decomposition:
\begin{equation}
\begin{split}
     \ket{\Psi}&= P_\epsilon \ket{\Psi} + (\mathds{1}-P_\epsilon) \ket{\Psi}\\
     & = \alpha \ket{a}+ \beta \ket{b}.
\end{split}
\end{equation}

The self adjoint rank-two operator $O_k= P_\epsilon\ket{\Psi}\bra{\Psi}P_\epsilon -\ket{\Psi}\bra{\Psi}$ , whose trace norm we want to evaluate, can now be represented as the following $2\times 2$ matrix acting on the basis $\{\ket{a},\ket{b}\}$ and as the zero operator everywhere else.

\begin{equation}
    P_\epsilon\ket{\Psi}\bra{\Psi}P_\epsilon -\ket{\Psi}\bra{\Psi}=\begin{pmatrix}
 0 & -\beta^* \alpha\\
 -\beta \alpha^* & -|\beta|^2
\end{pmatrix} \eqqcolon A
\end{equation}
Because the operator $O_k$ is self adjoint and because the complement of its kernel is the two dimensional subspace spanned by $\{\ket{a},\ket{b}\}$, it is block diagonal with one block being the above matrix $A$ and the other block being the zero operator.
Therefore its trace norm is equal to the trace norm of $A$.

The matrix $A$ has the eigenvalues:
\begin{equation}
    \lambda_{1/2}= -\frac{\beta}{2} \left( \beta \pm \sqrt{4-3\beta^2} \right).
\end{equation}
It is easy to see that for $\beta \in [0,1]$ the $\lambda_1<0$ as the expression in the bracket is always larger than 0. 
Since $\frac{\partial}{\partial \beta} \frac{\lambda_2}{\beta}= - \frac{1}{2} \left(1+\frac{3\beta}{\sqrt{4-3\beta^2}}\right)  <0$ and $\frac{\lambda_2}{\beta} \big|_{\beta = 1}= 0$ we can conclude by monotonicity that $\frac{\lambda_2}{\beta} \geq 0$ for $\beta \in [0,1]$.

This then gives us the trace norm:
\begin{equation}
\begin{split}
    \norm{A}_1 &= |\lambda_1|+|\lambda_2|\\
    &= \beta \sqrt{4-3\beta^2}
\end{split}
\end{equation}
Since $\sqrt{4-3\beta^2}\big|_{\beta=1}=1$ and $\sqrt{4-3\beta^2}\big|_{\beta=0}=2$ and $\frac{\partial}{\partial \beta} \sqrt{4-3\beta^2} =-\frac{3\beta}{\sqrt{4-3\beta^2}} <0$ we can conclude by monotonicity that
\begin{equation}
\begin{split}
    \norm{A}_1 \leq 2 \beta
\end{split}
\end{equation}
for $\beta \in [0,1]$.
Therefore, 
\begin{equation}
    \norm{O_k}_1 \leq 2 \beta
\end{equation}
holds, which completes the proof of the Lemma.
\end{proof}

Getting back to the main estimation, we arrive now at
\begin{equation}
    \norm{P_\epsilon \rho P_\epsilon- \rho }_1\leq \sum_k \, p_k \, 2 \, ||(1-P_\epsilon)\ket{\Psi_k}|| = 2 \sum_k \, p_k  \, \sqrt{||(1-P_\epsilon)\ket{\Psi_k}||^2 }
\end{equation}
We can now use the Jensen inequality in the form $\mathds{E}(\sqrt{X})\leq \sqrt{\mathds{E}(X)}$ for any random variable $X \geq 0$:
\begin{equation}
    \norm{P_\epsilon \rho P_\epsilon - \rho}_1\leq  2 \sqrt{\sum_k p_k   ||(1-P_\epsilon)\ket{\Psi_k}||^2 } = 2 \sqrt{\Tr((1-P_\epsilon){\rho})}
\end{equation}
By equation \eqref{Markov inequality applied} we can estimate the difference of the traces:
\begin{equation}
     \norm{P_\epsilon\rho P_\epsilon -\rho}_1\leq 2 \sqrt{\Tr(\rho)-\Tr(P_\epsilon \rho)} \leq 2 \sqrt{\frac{\epsilon^2}{4}} = \epsilon
\end{equation}
Therefore, the set $K$ is totally bounded.\\
Thus, $K$ is a convex compact set on a locally convex topological vector space. 
The quantum channel $S$ is a continuous function mapping from $K$ into itself.
Schauder-Tychonoffs fixed-point theorem can be applied. There exists a fixed point to the infinite-dimensional quantum channel $S$ in the set $K \subseteq D(\Hilbert)$, which completes the proof of Theorem \ref{thm1}.
\end{proof}

\bibliographystyle{plain} 
\bibliography{Bibliothek} 

\end{document}